\documentclass[prx,aps,showpacs,twocolumn,preprintnumbers,amsmath,amssymb,superscriptaddress,longbibliography]{revtex4-2}

\usepackage{graphicx}
\usepackage{amsmath}
\usepackage{amssymb}
\usepackage{amsthm}
\usepackage{comment}
\usepackage{mathtools}
\usepackage{xcolor}
\usepackage{hyperref}
\usepackage[normalem]{ulem}
\usepackage{simpler-wick}
\hypersetup{
  colorlinks=true,
  citecolor=blue,
  linkcolor=black,
  urlcolor=blue
}

\newcommand{\ket}[1]{\left| #1 \right\rangle}
\newcommand{\ketnolr}[1]{| #1 \rangle}
\newcommand{\bra}[1]{\left\langle #1 \right|}
\newcommand{\braket}[2]{\left\langle #1 \middle| #2 \right\rangle}
\newcommand{\ketbra}[2]{\left| #1 \middle\rangle \middle\langle #2 \right|}

\newcommand{\Tr}{\text{Tr}}

\newcommand{\im}{\iota}
\theoremstyle{definition}

\newtheorem{theorem}{Theorem}

\newcommand{\problem}[4]{\textbf{Problem}: #1 \newline
\textbf{Input}: #2 \newline \textbf{Output}: #4
\vspace{0.3cm}}

\begin{document}

\title{Quantum linear algebra for disordered electrons}
\author{Jielun Chen}
\affiliation{Division of Physics, Mathematics and Astronomy, California Institute of Technology, Pasadena, California 91125, USA}
\author{Garnet Kin-Lic Chan}
\affiliation{Division of Chemistry and Chemical Engineering, California Institute of Technology, Pasadena, California 91125, USA}

\begin{abstract}
We describe how to use quantum linear algebra to simulate a physically realistic model of disordered non-interacting electrons. The physics of disordered electrons outside of one dimension challenges classical computation due to the critical nature of the Anderson localization transition or the presence of large localization lengths, while the atypical distribution of the local density of states limits the power of disorder averaged approaches. Starting from the block-encoding of a disordered non-interacting Hamiltonian, we describe how to simulate key physical quantities, including the reduced density matrix, Green's function, and local density of states, as well as bulk-averaged observables such as the linear conductivity, using the quantum singular value transformation, quantum amplitude estimation, and trace estimation. We further discuss a quantum advantage that scales polynomially with system size and exponentially with lattice dimension.
\end{abstract}

\maketitle

\section{Introduction}
Quantum linear algebra provides a set of primitives for matrix computation on quantum computers \cite{Harrow_2009, Gilyen_2019, Martyn_2021}. For structured matrices with efficient block-encodings, this allows linear algebra to be performed on matrices with dimension $N=2^n$, i.e., exponential in the number of qubits $n$. In quantum simulation, most applications of quantum linear algebra focus on quantum many-body simulation, where the qubit Hilbert space represents the exponential many-body Hilbert space. Some recent works~\cite{babbush2023exponential,Barthe_2024,Somma_2024,stroeks2024solving}, however, have suggested the use of quantum linear algebra for simulations of an exponential number of non-interacting modes. A natural question in that context is whether there are interesting physical simulations which can be performed in this setting.

In the current work, we study one such potential application. Concretely, we consider the physics of non-interacting electrons (fermions) on a lattice. In perfect crystals, the translational symmetry permits a variety of efficient algorithms to obtain properties. However, a disordered material is one with an infinite-sized unit cell. Disorder may arise through impurities, either as intentional or non-intentional dopants~\cite{ziman1979models}, or through the structural parameters~\cite{lau2022reproducibility}. The pioneering work of Anderson \cite{anderson1958absence} showed that non-interacting electronic states are localized in the presence of sufficiently strong disorder, and a rich theory of Anderson localization has been developed in one, two, and three-dimensional lattices \cite{lee1985disordered}. However, the localization length is exponential in the conductance in two dimensions, while the localization transition is critical in three dimensions. Thus while numerics have provided much insight, the appearance of exponential length scales provides a challenge to determining exponents to high precision, as well as to studying systems where the underlying non-disordered lattice also has a large unit cell, as found in Moir\'{e} lattices. In addition, near the localization transition, the local density of states of a disordered material is not self-averaging; for example, the arithmetic disorder average does not serve as an order parameter~\cite{dobrosavljevic2003typical}. This limits the information provided by mean-field treatments based on disorder-averaged quantities. The ability to compute the properties of an actual realization of disorder on large system sizes is thus desirable.


We describe how to use quantum linear algebra to study disordered non-interacting electrons. We focus on uncorrelated disorder where each atom's data (e.g. type) is drawn independently from a distribution, and describe how to efficiently block encode a physically realistic disorder Hamiltonian via cryptographic tools that simulate random functions \cite{Zhandry_2012, Incompressibility, compressed_oracle}. Once this block encoding is computed, we apply the quantum singular value transform \cite{Gilyen_2019}
to compute the physical quantities, such as the reduced density matrix, single-particle Green's function, and local density of states. Finally, we also describe how to obtain bulk-averaged quantities, such as the linear conductivity, through quantum trace estimation.

We emphasize that the ability to simulate an exponentially large system does not imply exponential quantum advantage. In fact, for lattice systems where the Hamiltonian is geometrically local, as is the case for realistic disorder models, a class of classical algorithms can avoid exponential scaling for certain tasks. Additionally, to capture correlations at diverging length scales, the complexity of both quantum and classical algorithms may need to scale exponentially in $n$. However, a general polynomial advantage (with respect to classical algorithms we are aware of) can be established for realistic disorder models beyond 1D. We make this comparison towards the end of the paper.

There are prior works that are closely related to our presentation. For example, Ref.~\cite{alexandru2020quantum} highlighted the challenge of simulating critical localization transitions and subsequently described a Trotter-based algorithm for the time-evolution of a disorder Hamiltonian for exponentially large numbers of fermions. In the course of preparing our manuscript, the preprint~\cite{stroeks2024solving} appeared. While focused on the general setting of quantum linear algebra for an exponential number of fermion modes, Ref.~\cite{stroeks2024solving} also briefly discussed disorder Hamiltonians as an application enabled by the use of pseudorandom functions. Our work uses similar quantum linear algebra techniques to Ref.~\cite{stroeks2024solving}, but we obtain improved complexity estimates for these operations, and we provide the explicit block-encodings and oracle constructions for a physically realistic disorder model. 

\section{Non-interacting disordered systems}
Consider a system of $N=2^n$ atoms arranged on a crystalline lattice of dimension $D$, where each atom is labeled by index $i \in \{0,...,N-1\}$. For simplicity, we assume only one orbital per atom. 
This leads to the free-fermion Hamiltonian
\begin{equation}
    \hat{H} = \sum_{i=0}^{N-1} \sum_{j=0}^{N-1} h_{ij} c_i^\dagger c_j
\end{equation}
where $c_i/c_i^\dagger$ is the fermionic annihilation/creation operator. The $N \times N$ matrix $h$ is denoted the \emph{hopping matrix} (the diagonal terms $h_{ii}$ are sometimes called on-site energies, but we do not differentiate them from the off-diagonal terms here). 

Computation involving an arbitrary matrix $h$ can encode BQP-complete problems. However, here we are interested in matrix elements that arise in physical problems of disorder. To understand their form, imagine that each atom hosts an orbital that is exponentially localized around the atomic position, and define a basis from the set of all such orbitals $\{\phi_i\}$. Then $h_{ij}$ is the matrix element of a single-electron Hamiltonian $\hat{H}_0$ within the basis,  $h_{ij} = \int d \mathbf{r} \phi_i^*(\mathbf{r}) \hat{H}_0 \phi_j(\mathbf{r})$, 
and typically has the form \cite{Turchi_1998_tight}
\begin{equation}
    |h_{ij}| \sim e^{-\gamma |\mathbf{r}_i - \mathbf{r}_j|}
\end{equation}
where $\mathbf{r}_i$ is the coordinate of atom $i$, and $1/\gamma$ is the length-scale of the atomic orbital.
Since the hopping matrix elements decay rapidly with distance, the matrix $h$ can be truncated at finite length with exponentially small error, resulting in a sparse $h$ with each row containing only a few non-zero entries.
Furthermore, assuming $h_{ii}$ is bounded, the circle theorem implies that the spectral width of $h$ remains bounded independently of the system size $N$. 
An important point is that in a crystalline system, the elements $h_{ij}$ can be specified with $O(1)$ input data, as $\mathbf{r}_i$ is generated from the unit cell atomic positions and lattice vectors. 
Compact inputs can also specify other contributions to $h$, for example from external fields, or long wavelength deformations of the lattice~\cite{pereira2009tight}.


The corresponding disorder hopping matrix can be viewed as drawing entries of $h_{ij}$ from a certain random ensemble. A very common kind of disorder is when the atomic data can take multiple values associated with a probability distribution. For example, the type of atom can vary (modifying $\gamma$), the coordinate can be associated with a random displacement (modifying $\mathbf{r}_i$), or a random magnetic field could be present on each atom. We focus on these types of disorder in this work.

\section{Quantum linear algebra}
The physical properties of the free-fermion system are entirely encoded in matrix functions of the hopping matrix $h$, and computing such functions is a natural task for quantum linear algebra. We will employ block-encoding, the quantum singular value transform, and quantum amplitude estimation. An $\alpha$-scaled \emph{block-encoding} of $h$ is a unitary $U_h$ such that its submatrix is $h/\alpha$:
\begin{equation}
    U_h = 
    \begin{pmatrix}
        h/\alpha & *\\
        * &  *
    \end{pmatrix}
\end{equation}
where $\alpha$ is called the \emph{sub-normalization}, which ensures $\|h/\alpha\|_2 \leq 1$ such that the unitary exists. In our case, $h$ is a structured sparse matrix, thus $\alpha$ is simply the sparsity times $h$'s maximal element, and we can use sparse block-encoding techniques \cite{Gilyen_2019, Camps_2023}. 

The \emph{quantum singular value transform} (QSVT) constructs matrix functions given a block-encoded matrix. Specifically, given the block-encoding of $h/\alpha$, the QSVT constructs a new unitary $U^{p}_{h}$ such that
\begin{equation}
    U_{h}^{p} = 
    \begin{pmatrix}
        p(h/\alpha) & *\\
        * &  *
    \end{pmatrix}
\end{equation}
where $p$ is a degree-$d$ polynomial of the singular values of $h/\alpha$ (which are here eigenvalues since $h$ is Hermitian).
It suffices to prove that $p$ is uniformly bounded by a constant for QSVT to be efficiently applied \cite{Low_2016, Low_2017, Gilyen_2019, Martyn_2021, tang2023csguidequantumsingular, Motlagh_2024}, requiring only $O(d)$ queries to $U_h$ and $U_{h}^\dagger$ with a constant number of ancilla qubits. We highlight that to approximate $f(x) \approx p(x/\alpha)$, it is necessary to consider $f(x)$ on $[-\alpha, \alpha]$, which is then rescaled to $[-1,1]$, inducing overhead in two ways: first, the uniform norm of $f$ may scale with $\alpha$, which in the worst case induces an overhead of $\alpha^d$ in post-selection (e.g. for a monomial $x^d$); second, the rescaled $f$ is less smooth and requires a higher degree to approximate it. In our applications, these issues are negligible as the functions are bounded everywhere and the degree grows only linearly with $\alpha$.

The last step is to extract information from the block-encoded matrix function $f(h)$. This could be a bottleneck for some quantum algorithms since the amount of information required might scale with $N$. However, in the thermodynamic limit, one is only interested in intensive quantities, 
which makes extracting useful information possible. For example, we will be interested in extracting the matrix element $[f(h)]_{ij}$, which is the form of the local density of states or a local one-body observable, and which can be performed by \emph{quantum amplitude estimation} (QAE) \cite{Brassard_2002}. State-of-the-art variants of the algorithm can estimate such quantities within additive error $\epsilon$ with probability at least $1-\delta$ in $O(\frac{1}{\epsilon}\log(\frac{1}{\delta}))$ queries with various tradeoffs on depth, flexibility, etc. \cite{Aaronson_2020, Grinko_2021, Fukuzawa_2023, Rall_2023, Harrow_2020, Cornelissen_2023, QA_survey}. Additionally, we are interested in bulk-averaged quantities of the form $\Tr(f(h))/N$, such as the energy per site or the linear conductivity, which via the transformation
\begin{equation}
    \frac{\Tr(f(h))}{N} = \left(\frac{1}{\sqrt{N}}\sum_{i=0}^N \bra{i} \bra{i}\right) I \otimes f(h) \left(\frac{1}{\sqrt{N}}\sum_{j=0}^N \ket{j} \ket{j}\right)
\end{equation}
can again be obtained by applying QAE. We could also use stochastic trace estimation which can be more resource-efficient. There one draws $K$ random states $\ket{\psi_k}$ from an ensemble with the property $\mathbb{E}_k[\bra{\psi_k}f(h)\ket{\psi_k}] = \Tr(f(h))/N$, then one estimates $\bra{\psi_k}f(h)\ket{\psi_k}$
by QAE, and takes the mean to estimate the trace. Although this introduces an additional stochastic error, many good choices of ensemble $\{\ket{\psi_k}\}$, such as Hutchinson states \cite{Shen_2024} or 2-designs, cause the variance to vanish as $O(\|f(h)\|_F^2/N^2)$, and for any block-encodable matrix $f(h)$ we have $\|f(h)\|_F^2=O(N)$ or smaller, thus the stochastic error converges like $O(1/N)$, and for $N=2^n$ the error of QAE dominates the cost.

\section{Block-encoding the disordered hopping matrix}
We start by focusing on the binary alloy model \cite{Saha_1996}, which is a type of substitutional disorder where each atom randomly takes one of two atom types with probability $p$ and $1-p$ correspondingly. We then describe how to extend the same techniques to structural and magnetic disorders.

In the random binary alloy, for each $(i, j)$ the hopping matrix coefficients are now denoted as $h_{ij}^{ab}$, which takes four possible values depending on the atom types $a,b \in \{0,1\}$: $h_{ij}^{00}$, $h_{ij}^{11}$, and $h_{ij}^{01} = (h_{ij}^{10})^{*}$. There are $2^N$ total configurations of the atom type distribution, which we label as $A \in \{0,...,2^N-1\}$, and we associate each with a function $g_A: \{0,1\}^n \rightarrow \{0,1\}$ which maps the atom label to the atom type. We denote the hopping matrix specified by the configuration as $[h^{A}]_{ij} = h_{ij}^{g_A(i)g_A(j)}$. Suppose we encode all possible hopping terms into a $2N \times 2N$ matrix $\tilde{h}_{(i,a),(j,b)}$, then by defining the diagonal projector $P^A_{(i,a)} = \delta_{(i,a),(i,g_A(i))}$, we have
\begin{equation}
    [\tilde{h}^A]_{(i,a), (j,b)} = P^A_{(i,a)}~ \tilde{h}_{(i,a), (j,b)}~P^A_{(j,b)}
\end{equation}
which effectively projects $\tilde{h}$ to elements relevant to $h^A$. Then it is clear that any matrix power of $h^A$ can be expressed as
\begin{equation}
    \left(h^A\right)^d= 2\left( I \otimes \bra{+}\right)\left(\tilde{h}^A\right)^d\left(I \otimes \ket{+}\right)
\end{equation}
where $\ket{+} = (\ket{0} + \ket{1})/\sqrt{2}$. Thus for polynomial transformations, it is equivalent to consider the block-encoding of $\tilde{h}$ and $P^A$, up to a sub-normalization overhead of $2$.

We first describe the block-encoding of $\tilde{h}$ which encodes all possible hopping terms. As introduced earlier, we focus on the following form
\begin{equation}
    \tilde{h}_{(i,a),(j,b)} = t_{ab}e^{-{\gamma}_{ab} |\mathbf{r}_i - \mathbf{r}_j|} e^{\im \phi_{ij}}
\end{equation}
which is a sparse matrix with at most $2s$ elements on each row. 
Here, the parameters $t_{ab}$ and $\gamma_{ab}$ are often fit from experiments or ab initio calculations \cite{Turchi_1998_tight}. We also include the Peierls phase $\phi_{ij}$ which can account for 
magnetic fields which are sometimes studied together with disorder \cite{Fisher_1985}. We assume that the interatomic distances and the Peierls phases are efficiently computable, and that they can be specified with finite precision by $m$ bits, enabling the sparse block-encoding of $\tilde{h}$ \cite{Gilyen_2019, Camps_2023}. Specifically, consider an oracle, which when given a column index $j \in \{0,...,N-1\}$ and $l \in \{0,...,s-1\}$, computes:
\begin{equation}
    O_c:\ket{l}\ket{j} \rightarrow \ket{l}\ket{c(j,l)}
\end{equation}
where $c(j,l)$ is the index of the $l$-th sparse element in $j$-th row.
The efficient computation of interatomic distance and phase implies there exists two poly($n$)-time oracles
\begin{equation}
\begin{split}
    &O_d:\ket{l}\ket{j} \ket{0} \rightarrow \ket{l}\ket{j} \ket{\tilde{d}_{lj}}\\
    &O_\phi:\ket{l}\ket{j} \ket{0} \rightarrow \ket{l}\ket{j} \ket{\tilde{\phi}_{lj}}
\end{split}
\end{equation}
where $\tilde{d}_{lj}$ is the $m$-bit integer representation of $d_{c(j,l)j}=|\mathbf{r}_{c(j,l)} - \mathbf{r}_{j}|$, the distance between atom $j$ and atom $i = c(j,l)$. In other words, $\tilde{d}_{lj} = \left\lfloor M d_{c(j,l)j}/d_{\max}  \right\rfloor$, where $M=2^m$ and $d_{\max}$ is the upper bound on the maximal distance between two atoms. For convenience, we also introduce the rescaled quantity $\tilde{\gamma} = \gamma d_{\max}/M$ so that $\tilde{\gamma} \tilde{d}_{lj} = \gamma d_{c(j,l)j}$. Similarly, $\tilde{\phi}_{lj} = \left\lfloor M\phi_{c(j,l)j}/2\pi \right\rfloor$. Now we need to construct the oracle for the exponential distance factor and the phase oracle
\begin{equation}
\begin{split}
    &O_{e}:\ket{\tilde{d}_{lj}}\ket{0} \rightarrow \ket{\tilde{d}_{lj}} (t e^{-\tilde{\gamma} \tilde{d}_{lj}}\ket{0} + \sqrt{1-t^2e^{-2\tilde{\gamma} \tilde{d}_{lj}}}\ket{0^{\perp}})\\
    &O_{p}:\ket{\tilde{\phi}_{lj}} \rightarrow e^{\im 2\pi \tilde{\phi}_{lj}/M} \ket{\tilde{\phi}_{lj}} 
\end{split}
\end{equation}
where $\braket{0}{0^{\perp}} = 0$. These oracles allow us to block-encode the matrix $te^{-\gamma d_{ij}} e^{\im \phi_{ij}}$ using the circuit shown in Fig.~\ref{fig:exp_oracle} (a). For the text below we assume $|t|\leq1$, otherwise it enters into the sub-normalization as $|t|$. 

We now describe how to construct $O_{e}$ and $O_p$. In general, an oracle performing the classical computation of the matrix elements can be used to construct the block-encoding \cite{Sanders_2019}. Here, using the exponential form, we can significantly simplify the block-encoding. Consider the binary expansion
$\tilde{d}_{lj} = \sum_{k=0}^{m-1}2^{k} \tilde{d}_{lj}^k$ where $\tilde{d}_{lj}^k \in \{0,1\}$, then the exponential function can be decomposed as a product $e^{-\tilde{\gamma} \tilde{d}_{lj}} = \prod_{k=0}^{m-1} e^{-\tilde{\gamma} 2^{k} \tilde{d}_{lj}^k}$. The same applies to $e^{\im 2\pi \tilde{\phi}_{lj}/M}$. Using this property, we define the operator $A_k$ and its corresponding block-encoding $U_k$, and the phase gate $R_k$:
\begin{equation}
\begin{split}
&A_{k}=t^{1/m} 
\begin{pmatrix}
    1 & 0 \\
    0 & e^{-\tilde{\gamma} 2^{k}}
\end{pmatrix},
\quad
U_{k}=
\begin{pmatrix}
    A_{k } & * \\
    * & *
\end{pmatrix}\\
&R_k=
    \begin{pmatrix}
        1 & 0\\
        0 & e^{\im 2\pi 2^k/M}
    \end{pmatrix}
\end{split}
\end{equation}
Then $O_e = P^\dagger(\bigotimes_{k=0}^{m-1} U_{k})P$ where $P$ permutes $\ketnolr{\tilde{d}_{lj}}\ket{0^{\otimes m}}$ into $\ket{0}\ketnolr{\tilde{d}_{lj}^0}...\ket{0}\ketnolr{\tilde{d}_{lj}^{m-1}}$, and $O_p = \prod_{k=0}^{m-1} R_k$. In order to block-encode $t_{ab}e^{-\gamma_{ab} d_{ij}}$, we introduce two ancilla qubits acting as selectors, and construct $U_{\gamma,k}^{ab}$ which block-encodes
\begin{equation}
    A_{k}^{ab} = t_{ab}^{1/m} 
\begin{pmatrix}
    1 & 0 \\
    0 & e^{-\tilde{\gamma}_{ab} 2^{k}}
\end{pmatrix}.
\end{equation}
and use this to construct the new oracle $O_e^{ab}$. Then we use the circuit shown in Fig.~\ref{fig:exp_oracle} (b) for block-encoding, introducing an extra factor of $2$ to the sub-normalization. 

\begin{figure}
    \centering
    \includegraphics[width=1\linewidth]{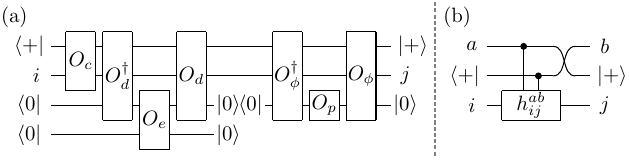}
    \caption{(The circuits read from right to left) (a) The full block encoding of an $s$-sparse matrix $te^{-\gamma d_{ij}}e^{\im \phi_{ij}}$, with sub-normalization $\alpha=s \max(1,|t|)$, where sparsity comes in because one inputs $\ket{+}=\sum_{l=0}^{s-1}\ket{l}/\sqrt{s}$. (b) Using two ancilla qubits to block-encode $h_{ij}^{ab}= t_{ab}e^{-\gamma_{ab} d_{ij}}e^{\im \phi_{ij}}$ with sub-normalization $\alpha=2 s \max(1,\max_{ab}|t_{ab}|)$.}
    \label{fig:exp_oracle}
\end{figure} 

Next, we discuss how to block-encode $P^A$ which specifies the disorder configuration. Directly defining $g_A$ coherently requires $O(N)$ amount of information, potentially diminishing the advantage of quantum algorithms. We describe two resolutions. Firstly, suppose $d = \text{poly}(n)$, then one can use (quantum-secure) pseudorandom functions (PRFs) \cite{GGM_1986, NR_1999, Zhandry_2012}. A PRF is a keyed function $f_k$, such that for a random key $k$ drawn from a key space $\mathcal{K}$, $f_k$ is indistinguishable from a uniformly random function, under $\text{poly}(n)$ (quantum) queries to its black-box. For simplicity, we assume input space = key space = output space, i.e $f_k:\{0,1\}^n\rightarrow \{0,1\}^n$ with $k \in \{0,...,N-1\}$. We define the corresponding reversible PRF as the mapping
\begin{equation}
    F_k: \ket{j}\ket{0} \rightarrow \ket{j}\ket{f_k(j)}.
\end{equation}
To indicate the atom type, we need to infer one bit from the output $f_k(j)$, such that the bit is $0$ with probability $1-p$ and $1$ with probability $p$. This can be implemented through inverse transform sampling. Suppose $f_k$ is truly random, then for each output $o = f_k(j)$, which uses the comparator circuit
\begin{equation}
    C_p: \ket{o}\ket{0} \rightarrow 
    \ket{o} \ket{c_p(o)}
\end{equation}
where $c_p(o) = 0$ if $o \ge  pN$ and $1 $ if $o < pN$. This can be implemented with $O(n)$ gates and ancillas \cite{Cuccaro_2004, Gidney_2018}. This effectively draws a random Boolean function taking $j$ as the input and with $c_p(f_k(j))$ as the output, with probability $p^w(1-p)^{N-w}$, where $w$ is the number of inputs such that the Boolean function's output is $1$. Since we used $f_k$ as a black box, this algorithm preserves the security of the PRF. To block-encode the projector, we compute $c_p(f_k(j)) \oplus b$ which can be achieved by a single CNOT. The full circuit for block-encoding is shown in Fig.~\ref{fig:PRF_H}. This method extends to any independent distribution with atom-specific probabilities $p_i$, as long as the oracle from $i$ to $p_i$ can be efficiently constructed.

However, suppose the polynomial degree required is superpolynomial in $n$, then the PRF may no longer be a reliable substitute for a random function. 
This is relevant if we wish to compute the tail of a correlation function such as $\langle c^\dag_i  c_j\rangle$ for systems with large localization lengths, and where we wish to consider $i$ and $j$ separated by a distance on the order of the system size. Then, to compute a non-trivial value for this element we would need to use $d = O(2^{\nu n})$ for some $\nu<1$. In this case, one can replace $f_k$ with $t$-wise independent functions, which are indistinguishable from truly random functions under $t$ queries. One can construct a $t$-wise independent function with $\text{poly}(n)\tilde{O}(t)$ random gates, each acting on a constant number of bits \cite{Incompressibility}. Thus an extra $\tilde{O}(d)$ multiplicative overhead is introduced to the time complexity. One can trade this for space complexity using a compressed oracle \cite{Zhandry_2012}, which constructs a perfect simulation of random functions under $d$ queries using $\tilde{O}(d)$ ancillas but only $\text{poly}(n)$ time. We emphasize that in this regime, the entire quantum algorithm scales superpolynomially in $n$, and thus cannot be viewed as an algorithm that treats exponentially many modes efficiently. However, the general polynomial advantage over classical algorithms still holds, which is discussed towards the end of the paper.

The above construction straightforwardly extends to encode structural and magnetic disorder. At the step of constructing $O_e$ in the block-encoding of $\tilde{h}$, one can modify the distance by a pseudorandom function, $t$-wise independent function, or compressed oracle. For example, for a 1D chain, one simply computes $\ket{x_i} \rightarrow \ket{x_i+f_k(i)}$ to account for a uniformly distributed random displacement around the lattice position (with $f_k$ modified to have $m'$-bit output where $m'$ is typically less than $m$ to control the width of the displacement. Non-uniform distributions can be again implemented through inverse transform sampling). The extension to 2D or 3D is straightforward. For random magnetic fields we consider the simplest model where $\phi_{ij}$ is uniformly drawn from $[0,2\pi)$ \cite{Avishai_1993}, thus in $O_p$ one can simply replace $\ketnolr{\tilde{\phi}_{lj}}$ with $\ket{f_k(sj+l)}$ where the input and output are modified to contain $n+\lceil \log(s)\rceil$ bits and $m$ bits respectively.

\begin{figure}[t]
    \centering
    \includegraphics[width=1\linewidth]{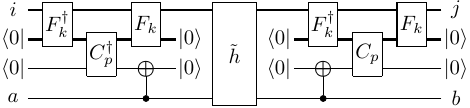}
    \caption{(The circuit reads from right to left) Block-encoding of a disordered Hamiltonian instance $\tilde{h}^k$ where $F_k$ is a pseudorandom function. If $d$ is superpolynomial in $n$, we replace $F_k$ with a $d$-wise independent function or compressed oracle.}
    \label{fig:PRF_H}
\end{figure}

\section{Estimating physical quantities}
We consider estimating physical observables with respect to the grand canonical ensemble density matrix defined as $\hat{\rho} = \exp(-\beta(\hat{H}-\mu\hat{N}))/Z$, where $\beta$ is the inverse temperature, $\mu$ is the chemical potential, $\hat{N}$ is the particle number operator, and $Z=\Tr(\exp(-\beta(\hat{H}-\mu\hat{N})))$ is the partition function. For the text below we simply write $h$ as the disordered hopping matrix in place of $\tilde{h}^A$ as introduced previously, to simplify notation. In general, the total cost for estimating physical quantities will be the product of the query complexities of constructing observables through QSVT, the time to block-encode $(h-\mu I)/\alpha$, and the QAE overhead $O(\frac{1}{\epsilon}\log(\frac{1}{\delta}))$, as well constants involving the sub-normalization. For clarity, we only present the quality of the polynomial approximation relevant to QSVT and procedures for estimating quantities. The total cost can be obtained from simple calculations, and we present proofs for degree bounds in Appendix.~\ref{app:chebyshev}. 

The static properties of a free fermion system at thermal equilibrium are fully specified by its one-body reduced density matrix (1-RDM) (sometimes referred to as the correlation matrix), defined as $D_{ij}= \Tr(c_j^\dagger c_i\hat{\rho})$. The Fermi-Dirac distribution expresses this as a matrix function of $h$:
\begin{equation}
    D_{ij} = \left[f_{\beta}(h-\mu I)\right]_{ij} = \left[\frac{1}{e^{\beta(h-\mu I)}+I}\right]_{ij}
\end{equation}
with $I$ being the identity. Given the block-encoding of $(h-\mu I)/\alpha$, one can use QSVT to block-encode a polynomial $p_f$ with degree $d_f = O(\alpha \beta \log(\frac{\alpha \beta}{\epsilon}))$ such that $\|p_f((h-\mu I)/\alpha) - D\|_2 \leq \epsilon$. Then any local one-body observable $\hat{O}= \sum_{ij}O_{ij}c_i^\dagger c_j$ can be computed via $\Tr(\hat{O}\hat{\rho}) = \sum_{ij}O_{ij} D_{ji}$. Since $\hat{O}$ is local, $O$ only contains a constant number of non-zero elements, thus one can estimate each of the  $D_{ji}$ individually and sum them together. 

An alternative approach is to encode observables in a quantum state using techniques like shadow Hamiltonian simulation \cite{Somma_2024}. For instance, the 1-RDM can be represented as $\ket{\psi_D} = \sum_{ij} D_{ij}/\|D\|_F \ket{i} \ket{j}$. However, the Frobenius-norm normalization scales with the particle number which is $O(N)$, while block-encodings normalize the operator by its sparsity. This means that when extracting a single element, the block-encoding strategy can be more efficient. Furthermore, given a block-encoding of $D$, $\ket{\psi_D}$ can be prepared with probabilities $\|D\|_F/N$. Thus for operators with extensive Frobenius-norm, as in our applications, the block-encoding we adopt is often more favorable.

Intensive bulk-averaged quantities $\Tr(OD)/N$, such as $O=h$ which gives the total energy per site, can 
be estimated 
by trace estimation techniques as described earlier. On the other hand, the dynamical properties of a free fermion system at thermal equilibrium can be fully captured by the single-particle Green's function. Among the variants, the retarded Green's function 
\begin{equation}
    G^R_{ij}(t) = -\im \theta(t) \Tr\left(\left\{c_i(t), c_j^\dagger (0)\right\} \hat{\rho} \right)
\end{equation}
is the simplest in the free-fermion setting, where $c_j(t) = e^{\im (\hat{H} - \mu \hat{N}) t}c_j e^{-\im (\hat{H} - \mu \hat{N}) t}$ is the annihilation operator in the Heisenberg picture, $\theta$ is the Heaviside step function, and $\{.,.\}$ denotes the anti-commutator. In the frequency domain, the retarded Green's function is the following  matrix function of $h$:
\begin{equation}
\begin{split}
    G^R_{ij}(\omega)&= \int_{-\infty}^{\infty} dt~e^{\im\omega t}e^{-\eta t} G^R_{ij}(t)\\
    &=\left[\frac{1}{(\omega+\im \eta)I-(h-\mu I)}\right]_{ij}
\end{split}
\end{equation}
where $\eta$ is the broadening parameter determining the resolution in the energy domain. In physical studies, finite experimental resolution implies $\eta$ is lower-bounded by a constant, thus the Chebyshev expansion converges quickly. Specifically, given the block-encoding of $(h-\mu I)/\alpha$, one can use QSVT to construct a polynomial $p_g$ with degree $d_g = O(\frac{\alpha}{\eta} \log(\frac{1}{\epsilon}))$ such that $\|p_g((h-\mu I)/\alpha) - \eta G^R(\omega)\|_2 \leq \epsilon$. We highlight the introduction of the sub-normalization $1/\eta$ to ensure that the 2-norm is bounded by $1$, which incurs overhead in the QAE.

In studies of Anderson localization, an important quantity is the imaginary part of the retarded Green's function, referred to as the spectral function $S_{ij}(\omega) = -\text{Im}G^R_{ij}(\omega)/\pi$, which is a Lorentzian function with $\eta$ specifying the width. The spectral function can be approximated by $\|-\text{Im}p_g((h-\mu I)/\alpha)-\pi\eta S(\omega)\|_2\leq\epsilon$. The real-space local density of states (LDOS) at atom $i$ is given by the diagonal elements  $S(\omega)_{ii}$ which can be estimated by QAE. The momentum space LDOS can also be estimated with the quantum Fourier transform (QFT) \cite{QFT} as $[\text{QFT}\cdot S(\omega)\cdot\text{QFT}^\dagger]_{kk}$.

In disordered materials, both the longitudinal and transverse (Hall) conductivities are relevant \cite{Lee_1985}. The static conductivity tensor in the linear response regime is computed via the Kubo-Bastin formula \cite{Bastin_1971, Streda_1982}:
\begin{equation}
    \sigma^{xy} = \frac{\im e^2 \hbar}{V_{at}N} \int d\omega f_{\beta}(\omega) \Tr \left[ v^x S(\omega) v^y \frac{\partial G^R(\omega)}{\partial\omega} - \text{h.c.}\right]
\end{equation}
where we note that the velocity operator matrix elements $[v^x]_{ij}$ also decay exponentially with distance $d_{ij}$, assuming exponentially localized bases. The block-encoding of the velocity operator thus resembles that of the Hamiltonian, but with sub-normalization $\alpha_v$ tied to sparsity and maximum velocity. One can evaluate the derivative by the identity $\partial G^R(\omega)/\partial\omega = -(G^R(\omega))^2$. The conductivity integral can be efficiently evaluated with Gauss quadrature, through the observation that the Chebyshev expansion of $G^R(\omega)$ as a polynomial of $h$ means it is also a polynomial in $\omega$, with the corresponding coefficients being matrix functions of $h$. A detailed discussion is in Appendix~\ref{app:quadrature}. 

\section{Discussions}
Our algorithms to obtain quantities of interest in non-interacting disordered systems inherit the advantages and caveats of quantum linear algebra. The techniques we use to estimate intensive quantities on lattices of size $N$ with additive error $\epsilon$ scale as $O(\text{polylog}(N)\frac{d}{\epsilon})$ up to a factor of sub-normalization, where $d$ is the sum of the polynomial degrees needed to approximate the involved operators. From the viewpoint of quantum linear algebra, our observables resemble the form of BQP-complete matrix functions \cite{Montanaro_2024, Santiago_2024}, and as a demonstration we prove that estimating the local density of states is BQP-complete in Appendix~\ref{app:BQP}. 

However, we emphasize that due to the geometrical locality of the Hamiltonians we study, this does not imply an exponential advantage. For the task of computing a single matrix element $\bra{i} f(h) \ket{j}$, and assuming $f$ can be approximated by a degree-$d$ polynomial, $f(h) \ket{j}$ is restricted to an effective light cone of volume $O(d^D)$ on a $D$-dimensional lattice. Existing classical algorithms we are aware of scale (at least) with this volume, whether using polynomial methods \cite{Wei_e_2006}, pole expansion and selected inversion \cite{pole, selinv}, or direct and iterative methods. While one can formulate random walk algorithms that do not iterate through the entire light cone, without non-trivial assumptions, their variance scales exponentially with the number of steps \cite{Montanaro_2024}.  Dequantization algorithms, such as quantum-inspired low-rank dequantization \cite{Chia_2020} or tensor-network-based dequantization \cite{Oseledets_2010, Oseledets_2011}, do not immediately apply either, due to the extensive Frobenius norm and incompressibility of a disordered instance. 
Directly comparing to classical kernel polynomial methods \cite{Wei_e_2006} that utilize the light cone, we obtain an $O(\epsilon d^D)$ advantage, since classically one has to expand the light cone $d$ times up to the maximal volume. On the other hand, when using physical hopping matrices and studying particular disorder phenomena for a given $\epsilon$, it may be possible to apply other classical estimation techniques adapted to the specific problem and accuracy requirement. 


Some particularly interesting directions to explore in the future include extensions to systems with correlated disorder, as well as the inclusion of interactions. In the latter case, the algorithms to compute Green's functions and traces provide the fundamentals for such developments.

\section{Acknowledgements}
We thank Tomislav Begu\v{s}i\'{c}, David Huse, Jiaqing Jiang, Chenghan Li, Urmila Mahadev, Sam McArdle, David Reichman, Yizhi Shen, and Samson Wang for insightful discussions. 

\begin{appendix}



\section{Polynomial approximations of observables}
\label{app:chebyshev}
In this appendix, we consider polynomial approximations of the Fermi-Dirac distribution and the retarded Green's function. In particular,  we consider error bounds on the Chebyshev expansion. Let $f$ be a function in the domain $[-1, 1]$, and suppose one can analytically extend it to the interior of the Berstein ellipse in the complex plane defined as
\begin{equation}
    E_{\rho} = \left\{ \frac{1}{2}(z + z^{-1}) : |z| \leq \rho \right\}
    \label{eq:ellipse}
\end{equation}
for some $\rho > 1$. Equivalently, this is the region where for $z= x+\im y$, $x$ and $y$ satisfy
\begin{equation}
    \left(\frac{x}{a}\right)^2 + \left(\frac{y}{b}\right)^2 \leq 1
\end{equation}
where
\begin{equation}
    a = \frac{\rho + \rho^{-1}}{2}, \quad b = \frac{\rho - \rho^{-1}}{2},
\end{equation}
\begin{theorem}
Suppose for all $z$ inside $E_\rho$, one has $|f(z)| \leq M$, then the $k$-th Chebyshev coefficient $a_k$ of $f(z)$ is bounded by $|a_k| \leq 2M\rho^{-k}$, and its degree-$d$ Chebyshev expansion $p$ satisfies
\begin{equation}
    \epsilon := \|f(x) - p(x)\|_{[-1,1]} \leq \frac{2M\rho^{-d}}{\rho - 1}.
\end{equation}
\label{theorem:chebyshev_ellipse}
\end{theorem}
\begin{proof}
    This is part of Theorems 8.1 and 8.2 from \cite{chebyshev_approx_bound}.
\end{proof}

We now prove the error bounds for Chebyshev expansions of relevant functions for the QSVT. First, we consider the Fermi-Dirac distribution.
\begin{theorem}
For the Fermi-Dirac distribution
\begin{equation}
    f(x) = \frac{1}{e^{\beta x}+1},
\end{equation}
defined on $x\in[-1,1]$, there is a polynomial $p_d$ of degree $d=O(\beta \log(\frac{\beta}{\epsilon}))$ such that $\|f - p_{d}\|_{[-1,1]} \leq \epsilon$ and $\|p_d\|_{[-1,1]}\leq 1+\epsilon$.
\end{theorem}
\begin{proof}
We extend $f$ analytically to the complex plane with a new variable $z=x+\im y$. The closest singularity for $f(z)$ is at $x=0$ and $y = \pm \pi/\beta$, therefore we can choose $\rho = 1+1/\beta$ so that for $x=0$, $|y| < 2/\beta <\pi/\beta$. Since $|\beta y| < 2$, we have $M\leq\max_{x} 1/(e^{\beta x}e^{2\im}+1) \leq 1/\sin(\pi-2) < 1.1$, then by Theorem.~\ref{theorem:chebyshev_ellipse} we have
\begin{equation}
\begin{split}
    \epsilon &\leq 2M\beta \left(1 + \frac{1}{\beta}\right)^{-d} \\
    d &\leq (\beta+1)\log\left(\frac{2M\beta}{\epsilon}\right)
\end{split}
\end{equation}
where we used $\log(1+1/\beta) > 1/(\beta+1)$. This concludes that $d = O(\beta \log(\frac{\beta}{\epsilon}))$.
\end{proof}
To cancel the sub-normalization $\alpha$ we therefore must approximate $1/(e^{\alpha \beta x}+1)$ instead, consequently, $d = O(\alpha\beta \log(\frac{\alpha\beta}{\epsilon}))$. 

We now consider the retarded Green's function. We note that near the boundary $\pm 1$, the function with half-width $\eta$ will lie outside of the window, potentially increasing the degree required. We can fix this by by increasing $\alpha$ slightly, e.g. replacing $\alpha$ with $2\alpha$, thus the spectrum will lie within some smaller region $[-c,c]$, e.g. $[-0.5,0.5]$ (thus also $\omega$).
\begin{theorem}
For the $\eta$-scaled retarded Green's function
\begin{equation}
    f(x) =  \frac{\eta}{(\omega + \im \eta)-x},
\end{equation}
defined on $x\in[-1,1]$, where $\omega \in [-c,c]$ for constant $c < 1$, there is a polynomial $p_d$ of degree $d=O(\frac{1}{\eta}\log(\frac{1}{ \epsilon}))$ such that $\|f - p_{d}\|_{[-1,1]} \leq \epsilon$ and $\|p_d\|_{[-1,1]}\leq 1+\epsilon$.    
\end{theorem}
\begin{proof}
Since this is a function of a simple pole, the Chebsyhev coefficients are analytically calculable \cite{Elliott1964TheEA} as
\begin{equation}
\begin{split}
    a_k &= \int_{-1}^1 dx \frac{T_k(x)}{\sqrt{1-x^2}} \frac{\eta}{(\omega + \im \eta)-x)}\\
    &= \frac{\pi\eta}{\sqrt{(\omega+\im \eta)^2-1}(\omega+\im \eta \pm \sqrt{(\omega+\im \eta)^2-1})^k}
\end{split}
\label{eq:GR_chebyshev_coeffs}
\end{equation}
where $\pm$ is chosen so that $|\omega+\im \eta \pm \sqrt{(\omega+\im \eta)^2-1} | > 1$. For $\omega \in [-c,c]$ for some $c < 1$, $|\sqrt{(\omega+\im \eta)^2-1}|$ is lower bounded by a constant independent of $\omega$ and $\eta$, and $|\omega+\im \eta \pm \sqrt{(\omega+\im \eta)^2-1}| \ge 1+\eta$. Then the worst-case scaling is $|a_k| = O(\eta (1+\eta)^{-k})$, which means $d = O(\frac{1}{\eta}\log(\frac{1}{ \epsilon}))$.
\end{proof}
Again to cancel the sub-normalization $\alpha$, we must consider
\begin{equation}
    \frac{\eta}{(\omega+ \im \eta) -\alpha x} = \frac{\eta/\alpha}{(\omega/\alpha+ \im \eta/\alpha) -x}
\end{equation}
and consequently $d = O(\frac{\alpha}{\eta}\log(\frac{1}{\epsilon}))$. 

\section{Gauss quadrature for conductivity}
\label{app:quadrature}
The Gauss quadrature considers the Lagrange interpolation on the function $f(x)$
\begin{equation}
    f(x) \approx p(x)= \sum_{j = 0}^d f(x_j) L_j(x)
\end{equation}
where $\{L_j\}$ are degree-$d$ polynomials such that $L_j(x_k) = \delta_{j, k}$, i.e., the Lagrange interpolating polynomials for a certain choice of grid points $\{x_j\}$. A common choice is the Chebyshev-Lobatto grid, defined as $x_j = \cos(j\pi/d)$ for $j=0,...,d$. Then the integral can be evaluated as
\begin{equation}
    \int_{-1}^{1} dx f(x) \approx \sum_{j=0}^d W_j f(x_j)
\end{equation}
where $W_j = \int_{-1}^{1} dx L_j(x)$ can be evaluated analytically for many choices of grids. The error of the integral is
\begin{equation}
\begin{split}
    \left|\int_{-1}^{1} dx f(x) - \sum_{j=0}^d W_j f(x_j) \right| &\leq \int_{-1}^1 \left|f(x) -p(x)\right|dx\\
    &\leq \|f(x)-p(x)\|_{[-1,1]} \int_{-1}^1 dx \\
    &= 2 \|f(x)-p(x)\|_{[-1,1]},
\end{split}
\end{equation}
thus as long as $p(x)$ is a good polynomial approximation to $f(x)$, the integration error is also small. We consider the error associated with the Chebyshev-Lobatto grids, which have an error bound very similar to the Chebyshev expansion through the following theorem
\begin{theorem}
Suppose for all $z$ inside $E_\rho$ defined in Eq.~\ref{eq:ellipse}, one has $|f(z)| \leq M$, then the degree-$d$ Lagrange interpolation $p(x)$ of $f(x)$ on Chebyshev-Lobatto on $\{x_j\}$ defined as $x_j = \cos(j\pi/d)$ satisfies
\begin{equation}
    \epsilon := \|f(x) - p(x)\|_{[-1,1]} \leq \frac{4M\rho^{-d}}{\rho - 1}.
\end{equation}
\label{theorem:chebyshev_ellipse_int}
\end{theorem}
\begin{proof}
    This is part of Theorem 8.2 from \cite{chebyshev_approx_bound}.
\end{proof}
Thus the degrees obtained in Section.~\ref{app:chebyshev} directly apply. In the context of the conductivity, we are interested in evaluating $\int \tilde{\sigma}(\omega) d\omega$ where
\begin{equation}
    \tilde{\sigma}(\omega) \propto  f_{\beta}(\omega)\frac{1}{N} \Tr \left[ v^x \text{Im}G^R(\omega) v^y  (G^R(\omega))^2\right].
\end{equation}
Each eigenvalue of $G^R$ can be viewed as a function of $\omega$, thus one can consider the matrix expansion
\begin{equation}
    \tilde{G}^R(\omega) = \sum_{j=0}^{d_g} G^R(\omega_j) L_j(\omega)
\end{equation}
where $d_g = O(\frac{1}{\eta}\log(\frac{1}{\eta\epsilon}))$, so that
\begin{equation}
    \max_{\omega} \|G^R(\omega) - \tilde{G}^R(\omega)\|_2 \leq \epsilon.
\end{equation}
If we substitute in $\tilde{\sigma}(\omega)$ with the polynomial approximations $p_f(\omega)\approx f_\beta(\omega)$ and $\tilde{G}^R(\omega)$, both with uniform error $\epsilon$, then the total error is
\begin{equation}
\begin{split}
    \epsilon_{total} = &\epsilon \frac{1}{N} \Tr \left[ v^x \text{Im}\tilde{G}^R(\omega) v^y  (\tilde{G}^R(\omega))^2\right] \\
    +& \epsilon p_f(\omega)\frac{1}{N} \Tr \left[ v^x v^y  (\tilde{G}^R(\omega))^2\right] \\
    +& 2\epsilon p_f(\omega)\frac{1}{N} \Tr \left[ v^x \text{Im}\tilde{G}^R(\omega) v^y  \tilde{G}^R(\omega)\right] \\
    +& ...
\end{split}
\end{equation}
where $\Tr(...)/N = O(1)$ for block-encodable matrices. Thus $\epsilon_{total} = O(\epsilon)$, suggesting the quadrature rule requires a similar amount of grid points to the QSVT query complexity. Note that in practice, one can consider alternatives of the Chebyshev-Lobatto grids to obtain the best convergence.

\section{BQP-completeness of estimating local density of states}
\label{app:BQP}
The $\pi \eta$-scaled spectral function is defined as
\begin{equation}
    S'(\omega,\eta,h) = -\eta \text{Im} G^R = \frac{\eta^2}{(\omega - h)^2+\eta^2}
\end{equation}
We consider the following problem of estimating the local density of states, which is of central interest in the study of Anderson localization.\\

\noindent
\problem{Estimating local density of states}{A $N \times N$ Hermitian matrix $h$ with norm $\|h\|_2 \leq 1$ with at most $4$ elements on each row and accessible through sparse access, a real number $\omega$, a positive number $\eta$, an index $j$, a precision $\epsilon$ and a threshold $g$, such that $\omega \in [-1,1]$, $\frac{1}{\eta} = O(\text{polylog}{(N)})$, $\frac{1}{\epsilon}=O(\text{polylog}{(N)})$, $g = O(1)$.}{$b$ is an upper bound on $\opNorm{A}$, $m$ and $\frac{1}{\varepsilon}$ are $\polylog{N}$.}{YES if $S'(\omega,\eta,h)_{jj} \geq g + \epsilon$, NO if $S'(\omega,\eta,h)_{jj} \leq g - \epsilon$.}

and we prove the following theorem.

\begin{theorem}
    The problem of estimating local density of states is BQP-complete.
\end{theorem}
\begin{proof}
The containment in BQP is implied by the algorithm in the main text. To prove BQP-hardness, we essentially adopt techniques from Ref. \cite{mono_BQP, Santiago_2024}. The starting point is to consider the following natural BQP-complete problem:\\

\noindent
\problem{BQP circuit simulation}{An $n$-bit string $x$ and a circuit $C = U_T \ldots U_1$, with $T=\text{poly}{(n)}$ and each $U_i$ acting non-trivially on at most 2 qubits, that acts on $r=\text{poly}{(n)} \geq n$ qubits as $C \ket{x} \ket{0}^{\otimes r-n} = \alpha_{x,0} \ket{0} \ket{\psi_{x,0}} + \alpha_{x,1} \ket{1} \ket{\psi_{x,1}}$, where $\ket{\psi_{x,0}}$, $\ket{\psi_{x,1}}$, $\alpha_{x,1}$, and $\alpha_{x,0}$ are all unknown except for the promise that either $|\alpha_{x,1}|^2 \geq \frac{2}{3}$ or $|\alpha_{x,1}|^2 \leq \frac{1}{3}$.}{Empty promise}{YES if $|\alpha_{x,1}|^2 \geq \frac{2}{3}$, NO if $|\alpha_{x,1}|^2 \leq \frac{1}{3}$.}

We use a standard method for proving BQP-hardness of matrix functions \cite{mono_BQP, Santiago_2024}. It aims to encode the circuit $C = U_T...U_1$ into an 4-sparse matrix $h$, so that reading out elements of some matrix function $f(h)_{jj}$ distinguishes between whether $|\alpha_{x,1}|^2 \geq \frac{2}{3}$ or $|\alpha_{x,1}|^2 \leq \frac{1}{3}$. We first define the extended circuit $C' = U_1^\dagger...U_T^\dagger (Z \otimes I^{r-1})U_T...U_1 = V_{M-1}...V_0$ where $M=2T+1$. We append $\lfloor \log M \rfloor + 1$ ancilla qubits and let the initial state be $\ket{s_{x}} = \ket{0}^{\lfloor \log M \rfloor + 1}  \ket{x} \ket{0}^{r-n}$, and define the corresponding clock unitary
\begin{equation}
    W = \sum_{l=0}^{M-1} \ketbra{l+1}{l} \otimes V_l
\end{equation}
with $\ketbra{l+1}{l}$ acting on $\ket{0}^{\lfloor \log M \rfloor + 1}$, and $\ket{M} = \ket{0}$. Our goal is to determine the spectrum of $W$. To do so, we define states $\ket{\phi_0^{b}}$ such that $C\ket{\phi_0^{b}} = \ket{b}\otimes\ket{\psi}$ for $b\in\{0,1\}$ for some arbitrary state $\ket{\psi}$, and for a sequence of states $\ket{\phi_k^b} = V_{k-1}...V_{0}\ket{\phi_0^b}$, one can verify that $V_{M-1}\ket{\phi_{M-1}^0}=\ket{\phi_0^0}$ and $V_{M-1}\ket{\phi_{M-1}^1}=-\ket{\phi_0^1}$. Then due to the cyclic structure, the eigenvalues and eigenstates are
\begin{equation}
\begin{split}
    &W\sum_{k=0}^{M-1}e^{-i2\pi lk/M}\ket{k}\otimes\ket{\phi_k^0}\\
    =& e^{i2\pi l/M}\sum_{k=0}^{M-1}e^{-i2\pi lk/M}\ket{k}\otimes\ket{\phi_k^0}
\end{split}
\end{equation}
for $l=0,...,M-1$, as well as
\begin{equation}
\begin{split}
    &W\sum_{k=0}^{M-1}e^{-i\pi (2l+1)k/M}\ket{k}\otimes\ket{\phi_k^1}\\ =& e^{i\pi(2l+1)/M}\sum_{k=0}^{M-1}e^{-i\pi (2l+1)k/M}\ket{k}\otimes\ket{\phi_k^1}.
\end{split}
\end{equation}
By choosing a set of complete orthogonal basis states $\{\ket{\psi}\}$ the eigenspace projector $P_l^+$ associated with eigenvalues $e^{i2\pi l/M}$ is
\begin{equation}
\begin{split}
    P_{l}^+= \frac{1}{M} \sum_{\psi}\sum_{k=0}^{M-1}\sum_{k'=0}^{M-1} e^{-i2\pi l(k-k')/M} \ketbra{k}{k'}\otimes\ketbra{\phi_k^0}{\phi_{k'}^0}.
\end{split}
\end{equation}
Then one can see that
\begin{equation}
\begin{split}
   \bra{s_x} P_l^+ \ket{s_x} =& \frac{1}{M}\sum_{\psi} \braket{s_x}{\phi_0^0}\braket{\phi_{0}^0}{s_x}\\
    =& \frac{1}{M}\sum_{\psi} \bra{s_x}C^\dagger (\ket{0}\otimes \ket{\psi})(\bra{0}\otimes \bra{\psi})C \ket{s_x}\\
    =& \frac{1}{M}\bra{s_x}C^\dagger (\ketbra{0}{0}\otimes I)C \ket{s_x}\\
    =& \frac{|\alpha_{0,x}|^2}{M}
\end{split}
\label{eq:P+}
\end{equation}
and similarly, the eigenspace $P_l^-$ associated with eigenvalues $e^{i\pi (2l+1)/M}$ has the property
\begin{equation}
    \bra{s_x} P_l^- \ket{s_x} = \frac{|\alpha_{1,x}|^2}{M}
\label{eq:P-}
\end{equation}
With this, we can define $h = (W+W^\dagger)/2$, which has at most $4$ elements on each row since $V_l$ belongs to a $2$-qubit universal gate set whose matrix has sparsity $2$. By the fact that 
\begin{equation}
\begin{split}
    &\frac{e^{i2\pi l/M} + e^{i2\pi (M-l)/M}}{2} = \cos\left(\frac{2\pi l}{M}\right) \\
    &\frac{e^{i\pi (2l+1)/M} + e^{i\pi (2M-1-2l)/M}}{2} = \cos\left(\frac{\pi (2l+1)}{M}\right)
\end{split}
\end{equation}
$h$ has eigenvalues $\pm \cos(2\pi l/M)$ for $l=0,...,(M-1)/2$, since
\begin{equation}
\begin{split}
    \cos\left(\frac{\pi(2l'+1)}{M}\right)&= \cos\left(\frac{\pi(2(\frac{M-1}{2}-l)+1)}{M}\right)\\ &= -\cos\left(\frac{2\pi l}{M}\right).
\end{split}
\end{equation}
Suppose $\ket{x}$ encodes the matrix index $j$, then $f(h)_{jj} = \bra{s_x}f(h)\ket{s_x}$. By grouping eigenspaces with associated eigenvalues and using Eq.~\ref{eq:P+} and Eq.~\ref{eq:P-}, we obtain
\begin{equation}
f(h)_{jj} = \frac{|\alpha_{x,0}|^2}{M}S^++\frac{|\alpha_{x,1}|^2}{M}S^-
\end{equation}
where $S^{\pm}$ are the sums
\begin{equation}
S^{\pm} = f(\pm1)+2\sum_{l=1}^{\frac{M-1}{2}}f\left(\pm\cos\left(\frac{2\pi l}{M}\right)\right).
\end{equation}
By substituting $|\alpha_{x,0}|^2 = 1- |\alpha_{x,1}|^2$ we have $f(h)_{jj} = S^+/M + |\alpha_{x,1}|^2(S^- - S^+)/M$. If we can distinguish $f(h)_{jj}$ for $|\alpha_{x,1}|^2 \ge 2/3$ and $|\alpha_{x,1}|^2 \leq 1/3$, extracting $f(h)_{jj}$ is BQP-hard. Now we consider the function $f(h) = S'(\omega, \eta, h)$. Our goal is to select $\omega$ as one of the eigenvalues in $S^-$ but not in $S^+$. By choosing a sufficiently small $\eta$, all other eigenvalues are evaluated to be close to $0$, ensuring that $S^+$ is small and $S^-$ is large. Due to the nature of the cosine function, we select $\omega$ as the eigenvalue closest to the origin. This choice ensures that the nearby eigenvalues are farther apart, minimizing their contributions to $S^+$. For even $T$, the eigenvalue closest to the origin (and presented in $S^-$) is $\omega = \omega^- := -\cos(\pi T / M)$, while for odd $T$, it is $\omega = \omega^+ := \cos(\pi T / M)$. Without loss of generality, we focus on the even $T$ case, as the odd case follows an identical argument. The closest eigenvalue has a distance of at least
\begin{equation}
\cos\left(\frac{\pi (T-1)}{M}\right) - \cos\left(\frac{\pi T}{M}\right) \geq \frac{\pi}{2M} =: \Delta.
\end{equation}
By choosing $\eta = c\Delta$ for some $c < 1$, and letting $x$ represent any eigenvalue other than $\omega^-$, the function $S'$ satisfies
\begin{equation}
S'(\omega^-, c\Delta, x) = \frac{c^2\Delta^2}{(\omega^- - x)^2 + c^2\Delta^2} \leq \frac{c^2}{1 + c^2} \leq c^2,
\end{equation}
while $S'(\omega^-, c\Delta, \omega^-) = 1$. This leads to bounds for $S^+$ and $S^-$:
\begin{equation}
0 \leq S^+ \leq Mc^2, \quad 2 \leq S^- \leq 2 + (M - 2)c^2.
\end{equation}
For the two cases $ |\alpha_{x,1}|^2 \geq 2/3 $ and $ |\alpha_{x,1}|^2 \leq 1/3 $, we have:
\begin{align}
&\frac{1}{M}S^+ + \frac{2}{3M}(S^- - S^+) \geq \frac{4}{3M} - \frac{2}{3}c^2,\\
&\frac{1}{M}S^+ + \frac{1}{3M}(S^- - S^+) \leq \frac{2}{3M} + \frac{2}{3}c^2.
\end{align}
By choosing $c \leq 1 / \sqrt{4M}$, we can distinguish between the two cases with $g = 1$ and an $\epsilon = 1 / (6M)$ estimation of $S'(\omega^-, c\Delta, h)_{jj}$.
\end{proof}

We note the techniques in \cite{Santiago_2024} allow one to determine the hardness of a wide range of matrix functions, for example, the Fermi-Dirac distribution can be analyzed similarly. 

\end{appendix}

\bibliography{references}

\end{document}